\documentclass[aip]{revtex4-1}
\usepackage{graphicx}
\usepackage{amssymb}
\usepackage{amsmath}
\usepackage{amsthm}
\usepackage{epsfig}
\usepackage{times}
\usepackage{subfigure}
\newtheorem{theorem}{Theorem}

\newtheorem{remark}{Remark}
\draft % marks overfull lines with a black rule on the right

\begin{document}

%\preprint{}

\title{Integrable semi-discretization of a multi-component short pulse equation}



\author{Bao-Feng Feng}
\affiliation{Department of Mathematics,
The University of Texas-Pan American,
Edinburg, TX 78541}
\author{Ken-ichi Maruno}
\affiliation{Department of Applied Mathematics,
Waseda University, Tokyo 169-8050, Japan}
\author{Yasuhiro Ohta}
\affiliation{Department of Mathematics,
Kobe University, Rokko, Kobe 657-8501, Japan}
%\thanks{Author to whom correspondence should be addressed.~Electronic mail:~feng@utpa.edu}

\begin{abstract}
 In the present paper, we mainly study the integrable semi-discretization of
a multi-component short pulse equation. Firstly, we briefly review the
bilinear equations for a multi-component short pulse equation proposed by
Matsuno (J. Math. Phys. \textbf{52} 123705) and reaffirm its $N$-soliton
solution in terms of pfaffians. Then by using a B\"{a}cklund transformation
of the bilinear equations and defining a discrete hodograph (reciprocal)
transformation, an integrable semi-discrete multi-component short pulse
equation is constructed. Meanwhile, its $N$-soliton solution in terms of
pfaffians is also proved.
\end{abstract}

\pacs{02.30.Ik, 05.45.Yv}% insert suggested PACS numbers in braces on next line

\keywords{multi-component short pulse equation; Hirota's bilinear method; integrable discretization;
Pfaffian solution;  complex short pulse equation}

\maketitle %\maketitle must follow title, authors, abstract and \pacs

\section{Introduction}
The nonlinear Schr\"{o}dinger (NLS) equation, as one of the universal
equations that describe the evolution of slowly varying packets of
quasi-monochromatic waves in weakly nonlinear dispersive media, has been
very successful in many applications such as nonlinear optics and water
waves \cite{Kodamabook, Agrawalbook,Boydbook,Yarivbook}.
%！
%The NLS equation is
%integrable, which can be solved by the inverse scattering transform \cite{Zakharov}.
The NLS equation is integrable and can be solved by the inverse scattering transform \cite{Zakharov}.
However, in the regime of ultra-short pulses where the width of optical
pulse is in the order of femtosecond ($10^{-15}$ s), the NLS equation
becomes less accurate \cite{Rothenberg}. Description of ultra-short
processes requires a modification of going beyond the standard slow varying
envelope approximation (SVEA) . Recently, Sch\"{a}fer and Wayne derived a
short pulse (SP) equation
\begin{equation}
u_{xt}=u+\frac{1}{6}(u^{3})_{xx}  \label{SPE}
\end{equation}%
in attempting to describe the propagation of ultra-short optical pulses in
nonlinear media \cite{SPE_Org}. Here, $u=u(x,t)$ is a real-valued function,
representing the magnitude of the electric field, the subscripts $t$ and $x$
denote partial differentiation. It has been shown that the SP equation
performs better than NLS under this case \cite{SPE_CJSW}.

Apart from the context of nonlinear optics, the SP equation has also been
derived as an integrable differential equation associated with
pseudospherical surfaces \cite{Robelo}. The SP equation has been shown to be
completely integrable \cite{Robelo,Beals,Sakovich,Brunelli1,Brunelli2}. The
loop soliton solutions as well as smooth soliton solutions of the SP
equation were found in \cite{Sakovich2,Kuetche}. The connection between the
SP equation and the sine-Gordon equation through the hodograph
transformation was clarified, and then the $N$-soliton solutions including
multi-loop and multi-breather ones were given in \cite%
{Matsuno_SPE,Matsuno_SPEreview} by using Hirota's bilniear method \cite{Hirota}.
%!
%The integrable discretizations
An integrable discretization
of the SP equation was constructed by the authors in \cite{SPE_discrete1}, and its geometric interpretation
was given in \cite{SPE_discrete2}.

A major simplification made in the derivation of the short pulse equation is
to assume that the polarization is preserved during its propagating inside
an optical fiber. However, this is not always the case in practice. For example,
we have to take into account the effects of polarization or anisotropy
in birefringent fibers \cite{Kartashov}.  Therefore, an extension to a
two-component version of the short pulse equation is needed in order to describe the
propagation of ultra-short pulse in birefringent fibers. In fact, several
integrable coupled short pulse have been proposed in the literature \cite%
{PKB_CSPE,Hoissen_CSPE,Matsuno_CSPE,Feng_CSPE,ZengYao_CSPE,comSPE}. The
bi-Hamiltonian structures for several coupled short pulse equations were
obtained in \cite{Brunelli_CSPE}.

In the present paper, we are concerned with the integrable semi-discretization of a multi-component short
pulse (MCSP) equation
\begin{equation}  \label{CSPE}
u_{i,xt} = u_i + \frac 12 \left( \left(\sum_{1 \le j < k \le n} c_{jk} u_j
u_k\right) u_{i,x} \right)_x\,, \quad i= 1, 2, \cdots, n\,,
\end{equation}
where the coefficients $c_{jk}$ are arbitrary constants with symmetry $%
c_{jk} = c_{kj}$. Eq. (\ref{CSPE}) was proposed by Matsuno \cite{Matsuno_CSPE} through Hirota's bilinear method, meanwhile, multi-soliton solution was given as well.

%！
%The reminder of the present paper is organized as follows. In Section 2, the
The remainder of the present paper is organized as follows. In Section 2, the
MCSP equation is briefly reviewed. We provide its $N$-soliton solution in an
%！
alternative pfaffian form and prove it by the
pfaffian technique. In Section 3, by using a B\"{a}cklund transformation of
the bilinear equations and defining a discrete hodograph transformation, we
construct a semi-discrete analogue of the MCSP equation.
Meanwhile, $N$-soliton solution in terms of pfaffian is provided and proved.
In Section 4, we investigate in detail the one- and two-soliton solutions to
the semi-discrete complex short pulse equation, which can be reduced from the MCSP equation.
The paper is concluded by several remarks in Section 5.

\section{Review of the multi-component short pulse equation and its
multi-soliton solution}

It was shown by Matsuno in \cite{Matsuno_SPE} that the SP equation (%
\ref{SPE}) is derived from bilinear equations
\begin{equation}
\left\{
\begin{array}{l}
\displaystyle D_{s}D_{y}\tilde{f}\cdot \tilde{f}=\frac{1}{2}\left( \tilde{f}%
^{2}-{\tilde{f^{\prime }}}^{2}\right) \,, \\[5pt]
\displaystyle D_{s}D_{y}\tilde{f}^{\prime }\cdot \tilde{f}^{\prime }=\frac{1%
}{2}\left( {\tilde{f^{\prime }}}^{2}-\tilde{f}^{2}\right) \,,%
\end{array}%
\right.   \label{SPE_bilinear1}
\end{equation}%
through transformations
\begin{equation}
u=2\mathrm{{i}}\left( \ln \frac{\tilde{f}^{\prime }}{\tilde{f}}\right)
_{s},\quad x=y-2(\ln \tilde{f}\tilde{f}^{\prime })_{s}\,,\quad t=s\,.
\label{hodograph1}
\end{equation}%
Here $D$ is called Hirota $D$-operator defined by
\[
D_{s}^{n}D_{y}^{m}f\cdot g=\left( \frac{\partial }{\partial s}-\frac{%
\partial }{\partial s^{\prime }}\right) ^{n}\left( \frac{\partial }{\partial
y}-\frac{\partial }{\partial y^{\prime }}\right) ^{m}f(y,s)g(y^{\prime
},s^{\prime })|_{y=y^{\prime },s=s^{\prime }}\,.
\]%
Recently, in view of the fact that the SP equation (\ref{SPE}) can also be
derived from another set of bilinear equations
\begin{equation}
\left\{
\begin{array}{l}
\displaystyle D_{s}D_{y}f\cdot g=fg\,, \\[5pt]
\displaystyle D_{s}^{2}f\cdot f=\frac{1}{2}g^{2}\,,%
\end{array}%
\right.   \label{SPE_bilinear2}
\end{equation}%
through transformations
\begin{equation}
u=\frac{g}{f},\quad x=y-2(\ln f)_{s}\,,\quad t=s\,,  \label{hodograph2}
\end{equation}%
Matsuno \cite{Matsuno_CSPE} constructed a multi-component generalization of the short pulse
equation (\ref{SPE}) based on a multi-component generalization of bilinear
equations (\ref{SPE_bilinear2}), which reads
\begin{equation}
\left\{
\begin{array}{l}
\displaystyle D_{s}D_{y}f\cdot g_{i}=fg_{i}\,,\quad i=1,2,\cdots ,n\,, \\%
[5pt]
\displaystyle D_{s}^{2}f\cdot f=\frac{1}{2}\sum_{1\leq j<k\leq
n}c_{jk}g_{j}g_{k}\,.%
\end{array}%
\right.   \label{CSP_Bilinear}
\end{equation}%
%!

\begin{remark}
The set of bilinear equations (\ref{SPE_bilinear1}) is actually obtained
from a 2-reduction of the KP-Toda hierarchy,
which basically delivers only two tau-functions out of a sequence of the tau-functions.  Furthermore, when these
two tau-functions are made complex conjugate to each other, the bilinear
equations (\ref{SPE_bilinear1}) is converted into the sine-Gordon equation $%
\phi _{ys}=\sin \phi $ via a transformation $\phi =2\mathrm{{i}\ln ({\tilde{f%
}^{\prime }}/{\tilde{f}})}$, which is further converted into the SP equation (\ref{SPE})
by a hodograph transformation.
\end{remark}

\begin{remark}
In \cite{Hirota_Ohta_sG}, Hirota and one of the authors have shown that both
the bilinear equations (\ref{SPE_bilinear1}) and (\ref{SPE_bilinear2})
derive the sine-Gordon equation. Furthermore, the relations between tau-functions, which read
\begin{equation}
f= {\tilde f^{\prime }}{\tilde f}, \quad g=2\mathrm{{i}} D_s {\tilde f^{\prime }}%
\cdot {\tilde f}\,,
\end{equation}
were also presented.
\end{remark}

\begin{remark}
As mentioned previously, Eqs. (\ref{SPE_bilinear1}) originate from 2-reduction of single component KP-Toda hierarchy, whereas, Eqs. (\ref{SPE_bilinear2}) come from $(1+1)$-reduction of two-component KP-Toda hierarchy. Since they both belong to $A^{(1)}_1$ of the Euclidean Lie algebra \cite{JM},  it is natural that they both derive the SP equation. However, the latter can be easily extended to $(1+ \cdots + 1)$-reduction of multi-component KP-Toda hierarchy, which gives rise to the multi-component generalization of the short pulse equation.
\end{remark}

In what follows, we will briefly review how the bilinear equations (\ref%
{CSP_Bilinear})
%!
%derives a multi-component generalization of the short pulse
determines a multi-component generalization of the SP
equation(\ref{SPE}). Dividing both sides  by $f^{2}$, the bilinear
equations (\ref{CSP_Bilinear}) can be cast into
\begin{equation}
\left\{
\begin{array}{l}
\displaystyle\left( \frac{g_{i}}{f}\right) _{sy}+2\frac{g_{i}}{f}\left( \ln
f\right) _{sy}=\frac{g_{i}}{f}\,,\quad i=1,2,\cdots ,n\,, \\[5pt]
\displaystyle\left( \ln f\right) _{ss}=\frac{1}{4}\sum_{1\leq j<k\leq
n}c_{jk}\frac{g_{j}g_{k}}{f^{2}}\,.%
\end{array}%
\right.   \label{CSP_BL2}
\end{equation}%
Introducing a hodograph transformation
\begin{equation}
x=y-2(\ln f)_{s}\,,\quad t=s\,,  \label{CSP_hodograph}
\end{equation}%
and a dependent variable transformation
\begin{equation}
u_{i}=\frac{g_{i}}{f}\,,\quad (i=1,2,\cdots ,n),  \label{CSP_transformation}
\end{equation}%
we then have
\[
\frac{\partial x}{\partial s}=-2(\ln f)_{ss}=-\frac{1}{2}\sum_{1\leq j<k\leq
n}c_{jk}u_{j}u_{k}\,,\qquad \frac{\partial x}{\partial y}=1-2(\ln f)_{sy}\,,
\]%
which implies
\begin{equation}\label{cov-relation}
  {\partial _{y}}=\rho ^{-1}{\partial _{x}}\,,\qquad {\partial _{s}}={\partial_t}-\frac{1}{2}\left( \sum_{1\leq j<k\leq n}c_{jk}u_{j}u_{k}\right) {\partial
_{x}}\,
\end{equation}
by letting $1-2(\ln f)_{sy}=\rho ^{-1}$.

Notice that the first equation in (\ref{CSP_BL2}) can be rewritten as
\[
\left(\frac{g_i}{f} \right)_{sy} = \left(1-2(\ln f)_{sy} \right) \frac{g_i}{f%
}\,,
\]
or
\[
\rho \left(\frac{g_i}{f} \right)_{sy} = \frac{g_i}{f}\,,
\]
which is converted into
\begin{equation}  \label{CPE1}
\partial_x \left(\partial_t - \frac 12 \left(\sum_{1 \le j < k \le n} c_{jk}
u_j u_k\right)\partial_x \right)u_i = u_i\,,
\end{equation}
by the conversion relation (\ref{cov-relation}). Obviously, Eq. (\ref{CPE1}) is nothing but the MCSP (\ref{CSPE}).

Next, we give an alternative representation of the $N$-solution to
%the MCSP equation (\ref{CSPE}) in the form of pfaffian. To this end,
the MCSP equation (\ref{CSPE}) in the form of pfaffians. To this end,
let us define a class of set $B_{\mu }$, $\mu =1,2,\cdots ,n$, which
satisfies the following condition,
\[
B_{\mu }\cap B_{\nu }=\emptyset ,\ \ \mathrm{if}\ \ \mu \neq \nu ,\qquad
\cup _{\mu =1}^{n}B_{\mu }=\{b_{1},b_{2},\cdots ,b_{2N}\}.
\]%
Then we define the elements of the pfaffians (others not mentioned below are
zeros)
\begin{equation}
\mathrm{Pf}(a_{j},a_{k})=\frac{p_{j}-p_{k}}{p_{j}+p_{k}}e^{\xi _{j}+\xi
_{k}}\,,\quad \mathrm{Pf}(a_{j},b_{k})=\delta _{j,k}\,,  \label{CSPE_pf1}
\end{equation}%
\begin{equation}
\mathrm{Pf}(b_{j},b_{k})=\frac{1}{4}\frac{c_{\mu \nu }}{p_{j}^{-2}-p_{k}^{-2}%
}\,,\quad (b_{j}\in B_{\mu },b_{k}\in B_{\nu })\,,  \label{CSPE_pf2}
\end{equation}%
\begin{equation}
\mathrm{Pf}(d_{l},a_{k})=p_{k}^{l}e^{\xi _{k}}\,,  \label{CSPE_pf3}
\end{equation}

\begin{equation}  \label{CSPE_pf4}
\mathrm{Pf}(b_j,\beta_\mu)= \left\{
\begin{array}{ll}
1 & \quad b_j \in B_\mu \\
0 & \quad b_j \notin B_\mu%
\end{array}%
\right.\,.
\end{equation}

%\begin{equation} \label{CSPE_pf5}
%{\rm Pf}(d_l,\beta_\mu)= {\rm Pf}(d_l,b_k)= {\rm Pf}(a_j,\beta_\mu) =0\,.
%\end{equation}
Here $j,k=1, 2, \cdots, 2N$, $\mu, \nu= 1, 2, \cdots, n$, $\xi_j=p_j y +
p_j^{-1} s + \xi_{i0}$ and $l$ is an integer. By defining the elements of the pfaffians, we can give the pfaffian solutions satisfying bilinear equations (\ref{CSP_Bilinear}).

\begin{theorem}
The bilinear equations (\ref{CSP_Bilinear}) have the following pfaffian solution
\begin{eqnarray}
f &=& \mathrm{Pf} (a_1, \cdots, a_{2N}, b_1, \cdots, b_{2N})\,, \\
g_i &=& \mathrm{Pf} (d_0, \beta_i, a_1, \cdots, a_{2N}, b_1, \cdots,
b_{2N})\,,
\end{eqnarray}
where $i=1,2, \cdots, n$ and the elements of the pfaffians are given in Eqs.
(\ref{CSPE_pf1})--(\ref{CSPE_pf4}).
\end{theorem}

\begin{proof}
Since
\[
\frac{\partial} {\partial y} \mathrm{Pf} (a_j,a_k)= (p_j -
p_k)e^{\xi_j+\xi_k} =  \mathrm{Pf} (d_0, d_1, a_j,a_k)\,,
\]

\[
\frac{\partial} {\partial s} \mathrm{Pf} (a_j,a_k)= (p^{-1}_k - p^{-1}_j)
e^{\xi_j+\xi_k}  = \mathrm{Pf} (d_{-1}, d_0, a_j,a_k)\,,
\]

\[
\frac{\partial^2} {\partial s^2} \mathrm{Pf} (a_j,a_k)= (p^{-2}_k -
p^{-2}_j) e^{\xi_j+\xi_k}  = \mathrm{Pf} (d_{-2}, d_0, a_j,a_k)\,,
\]
\[
\frac{\partial^2} {\partial y \partial s}\mathrm{Pf} (a_j,a_k)= (p_jp^{-1}_k
- p_k p^{-1}_j) e^{\xi_j+\xi_k} =  \mathrm{Pf} (d_{-1}, d_1, a_i,a_j)\,,
\]
where $\mathrm{Pf} (d_{l}, d_m)=0$ for integers $l$ and $m$, we then have
\[
\frac{\partial f} {\partial y} = \mathrm{Pf} (d_0, d_1, \cdots)\,,
\]

\[
\frac{\partial f} {\partial s} = \mathrm{Pf} (d_{-1}, d_0, \cdots)\,,
\]

\[
\frac{\partial^2 f} {\partial s^2} = \mathrm{Pf} (d_{-2}, d_0, \cdots)\,,
\]

\[
\frac{\partial^2 f} {\partial y \partial s} = \mathrm{Pf} (d_{-1}, d_1,
\cdots)\,.
\]
Here $\mathrm{Pf} (d_0, d_1, a_1, \cdots, a_{2N}, b_1, \cdots, b_{2N})$ is
abbreviated by $\mathrm{Pf} (d_0, d_1, \cdots)$, so as other similar
pfaffians.

Furthermore, it can be shown
\begin{eqnarray*}
&& \frac{\partial g_i} {\partial y} = \frac{\partial} {\partial y} \left[%
\sum_{j=1}^{2N} (-1)^{j} \mathrm{Pf} (d_0, a_j) \mathrm{Pf} (\beta_i,
\cdots ,\hat{a}_j, \cdots)\right] \\
&& =\sum_{j=1}^{2N} (-1)^{j} \left[ \left( {\partial_y} \mathrm{Pf} (d_0,
a_j) \right) \mathrm{Pf} (\beta_i, \cdots ,\hat{a}_j, \cdots) + \mathrm{Pf}
(d_0, a_j) {\partial_y} \mathrm{Pf} (\beta_i, \cdots ,\hat{a}_j, \cdots) %
\right] \\
&& =\sum_{j=1}^{2N} (-1)^{j} \left[ \mathrm{Pf} (d_1, a_j) \mathrm{Pf}
(\beta_i, \cdots ,\hat{a}_j, \cdots) + \mathrm{Pf} (d_0, a_j) \mathrm{Pf}
(\beta_i, d_0, d_1, \cdots ,\hat{a}_j, \cdots) \right] \\
&& = \mathrm{Pf} (d_1, \beta_i, \cdots)+ \mathrm{Pf} ( d_0, \beta_i, d_0,
d_1, \cdots) \\
&& = \mathrm{Pf} (d_1, \beta_i,  \cdots)\,.
\end{eqnarray*}

Here $\hat{a}_j$ means that the index $j$ is omitted. Similarly, we can show
\[
\frac{\partial g_i} {\partial s} = \mathrm{Pf} (d_{-1}, \beta_i, \cdots)\,,
\]

\begin{eqnarray*}
&& \frac{\partial^2 g_i} {\partial y \partial s} = \frac{\partial} {\partial
y} \left[\sum_{j=1}^{2N} (-1)^{j} \mathrm{Pf} (d_{-1}, a_j) \mathrm{Pf}
(\beta_i, \cdots ,\hat{a}_j, \cdots)\right] \\
&& =\sum_{j=1}^{2N} (-1)^{j} \left[ \left( {\partial_y} \mathrm{Pf}
(d_{-1}, a_j) \right) \mathrm{Pf} (\beta_i, \cdots ,\hat{a}_j, \cdots) +
\mathrm{Pf} (d_{-1}, a_j) {\partial_y} \mathrm{Pf} (\beta_i, \cdots ,\hat{a}%
_j, \cdots) \right] \\
&& =\sum_{j=1}^{2N} (-1)^{j} \left[ \mathrm{Pf} (d_0, a_j) \mathrm{Pf}
(\beta_i, \cdots ,\hat{a}_j, \cdots) + \mathrm{Pf} (d_{-1}, a_j) \mathrm{Pf}
(\beta_i, d_0, d_1, \cdots ,\hat{a}_j, \cdots) \right] \\
&& = \mathrm{Pf} ( d_0, \beta_i,\cdots)+ \mathrm{Pf} (d_{-1}, \beta_i, d_0,
d_1, \cdots)\,.
\end{eqnarray*}

An algebraic identity of pfaffian \cite{Hirota}
\begin{eqnarray*}
&& \mathrm{Pf} (d_{-1}, \beta_i, d_0, d_1, \cdots) \mathrm{Pf} (\cdots)=
\mathrm{Pf} (d_{-1}, d_0, \cdots) \mathrm{Pf} (d_1, \beta_i,  \cdots) \\
&& \quad - \mathrm{Pf} (d_{-1}, d_1, \cdots) \mathrm{Pf} ( d_0, \beta_i,
\cdots) + \mathrm{Pf} ( d_{-1}, \beta_i, \cdots) \mathrm{Pf} (d_0, d_1,
\cdots)\,,
\end{eqnarray*}
implies
\[
( {\partial_s} {\partial_y} g_i-g_i) \times f = {\partial_s} f \times {%
\partial_y} g_i - {\partial_s} {\partial_y} f \times g_i +  {\partial_s} g_i
\times {\partial_y} f \,,
\]
which is actually the first bilinear equation.

%!
%bilinear equation can be proved in a similar way by Iwao and
%Hirota \cite{IwaoHirota}.
The second bilinear equation is proved in a similar way as the one used by Iwao and
Hirota \cite{IwaoHirota} in connection with a different system. We start from the r.h.s of the bilinear equation.
\begin{eqnarray}  \label{CSP1_proof1}
&& \frac 12 \sum_{1 \le \mu < \nu \le n} c_{\mu \nu} g_{\mu} g_{\nu}
\nonumber \\
&& = \frac 14 \sum_{1 \le \mu,\nu \le n} c_{\mu \nu} \mathrm{Pf} (d_0, \beta_\mu,
 \cdots) \mathrm{Pf} (d_0, \beta_\nu,  \cdots)  \nonumber \\
&& = \frac 14 \sum_{1 \le \mu,\nu \le n} c_{\mu \nu } \sum_{i,j}^{2N}
(-1)^{i+j}\mathrm{Pf} (\beta_\mu, b_i) \mathrm{Pf} (d_0, \cdots,\hat{b}_i,
\cdots) \mathrm{Pf} (\beta_\nu, b_j) \mathrm{Pf} (d_0, \cdots,\hat{b}_j,
\cdots)  \nonumber \\
&& = \sum_{i,j}^{2N} (-1)^{i+j} \sum_{1 \le \mu,\nu \le n} \frac 14 c_{\mu
\nu } \mathrm{Pf}(\beta_\mu, b_i) \mathrm{Pf}(\beta_\nu, b_j) \mathrm{Pf}%
(d_0, \cdots,\hat{b}_i, \cdots) \mathrm{Pf} (d_0, \cdots,\hat{b}_j, \cdots)
\nonumber \\
&& = \sum_{i,j}^{2N} (-1)^{i+j} \left(p_i^{-2}-p_{j}^{-2}\right) \mathrm{Pf}
(b_i,b_j)\mathrm{Pf}(d_0, \cdots,\hat{b}_i, \cdots) \mathrm{Pf} (d_0, \cdots,%
\hat{b}_j, \cdots)  \nonumber \\
\end{eqnarray}
Next, the expansion of the vanishing pfaffian $\mathrm{Pf} (b_i, d_0, \cdots)
$ on $b_i$ yields
\[
\sum_{j=1}^{2N} (-1)^{i+j}\mathrm{Pf} (b_i, b_j) \mathrm{Pf} (d_0, \cdots,
\hat{b}_j, \cdots) = \mathrm{Pf} (d_0, \cdots,\hat{a}_i, \cdots)\,,
\]
which subsequently leads to
\begin{eqnarray}
&& \sum_{i,j}^{2N} (-1)^{i+j} p_i^{-2} \mathrm{Pf} (b_i,b_j)\mathrm{Pf}(d_0,
\cdots,\hat{b}_i, \cdots) \mathrm{Pf} (d_0, \cdots,\hat{b}_j, \cdots)
\nonumber \\
&& = \sum_{i}^{2N} p_i^{-2} \mathrm{Pf} (d_0, \cdots,\hat{a}_i, \cdots)
\mathrm{Pf} (d_0, \cdots,\hat{b}_i, \cdots)\,.  \label{CSP1_proof2}
\end{eqnarray}
Similarly, we can show
\begin{eqnarray}
&& -\sum_{i,j}^{2N} (-1)^{i+j} p_j^{-2} \mathrm{Pf} (b_i,b_j)\mathrm{Pf}%
(d_0, \cdots,\hat{b}_i, \cdots) \mathrm{Pf} (d_0, \cdots,\hat{b}_j, \cdots)
\nonumber \\
&& = \sum_{j}^{2N} p_j^{-2} \mathrm{Pf} (d_0, \cdots,\hat{a}_j, \cdots)
\mathrm{Pf} (d_0, \cdots,\hat{b}_j, \cdots)\,.  \label{CSP1_proof3}
\end{eqnarray}
Substituting Eqs. (\ref{CSP1_proof2})--(\ref{CSP1_proof2}) into Eq. (\ref%
{CSP1_proof1}), we arrive at
\begin{equation}  \label{CSP1_proof4}
\frac 12 \sum_{1 \le \mu < \nu \le n} c_{\mu \nu} g_{\mu} g_{\nu} =
2\sum_{i}^{2N} p_i^{-2} \mathrm{Pf} (d_0, \cdots,\hat{a}_i, \cdots) \mathrm{%
Pf} (d_0, \cdots,\hat{b}_i, \cdots)\,.
\end{equation}
Now we work on the l.h.s. of the second bilinear equation
\begin{eqnarray}
&& \frac{\partial^2 f} {\partial s^2} \times 0 - \frac{\partial f} {\partial
s} \frac{\partial f} {\partial s}  \nonumber \\
&& = \mathrm{Pf} (d_{-2}, d_0, \cdots) \mathrm{Pf} (d_{0}, d_0, \cdots) -
\mathrm{Pf} (d_{-1}, d_0, \cdots) \mathrm{Pf} (d_{-1}, d_0, \cdots)
\nonumber \\
&& = \sum_{i=1}^{2N} (-1)^i \mathrm{Pf} (d_{-2}, a_i) \mathrm{Pf} (d_0,
\cdots, \hat{a}_i, \cdots) \sum_{j=1}^{2N} (-1)^j \mathrm{Pf} (d_{0}, a_j)
\mathrm{Pf} (d_0, \cdots, \hat{a}_j, \cdots)  \nonumber \\
&& -\sum_{i=1}^{2N} (-1)^i \mathrm{Pf} (d_{-1}, a_i) \mathrm{Pf} (d_0,
\cdots, \hat{a}_i, \cdots) \sum_{j=1}^{2N} (-1)^j \mathrm{Pf} (d_{-1}, a_j)
\mathrm{Pf} (d_0, \cdots, \hat{a}_j, \cdots)  \nonumber \\
&& =\sum_{i,j=1}^{2N} (-1)^{i+j} \left[ \mathrm{Pf} (d_{-2}, a_i) \mathrm{Pf}
(d_{0}, a_j) -\mathrm{Pf} (d_{-1}, a_i) \mathrm{Pf} (d_{-1}, a_j) \right]
\nonumber \\
&& \quad \times \mathrm{Pf} (d_0, \cdots, \hat{a}_i, \cdots) \mathrm{Pf}
(d_0, \cdots, \hat{a}_j, \cdots)  \nonumber \\
&&=\sum_{i,j=1}^{2N} (-1)^{i+j+1} \left[p_i^{-2} + p_i^{-1}p_j^{-1} \right]
\mathrm{Pf} (a_i, a_j) \mathrm{Pf} (d_0, \cdots, \hat{a}_i, \cdots) \mathrm{%
Pf} (d_0, \cdots, \hat{a}_j, \cdots)   \nonumber
\end{eqnarray}
The summation over the second term within the bracket vanishes due to the
fact that
\begin{eqnarray*}
&& \sum_{i,j=1}^{2N} (-1)^{i+j+1} p_i^{-1}p_j^{-1} \mathrm{Pf} (a_i, a_j)
\mathrm{Pf} (d_0, \cdots, \hat{a}_i, \cdots) \mathrm{Pf} (d_0, \cdots, \hat{a%
}_j, \cdots) \\
&& = \sum_{j,i=1}^{2N} (-1)^{j+i+1} p_j^{-1}p_i^{-1} \mathrm{Pf} (a_j, a_i)
\mathrm{Pf} (d_0, \cdots, \hat{a}_j, \cdots) \mathrm{Pf} (d_0, \cdots, \hat{a%
}_i, \cdots) \\
&& = -\sum_{i,j=1}^{2N} (-1)^{i+j+1} p_i^{-1}p_j^{-1} \mathrm{Pf} (a_i, a_j)
\mathrm{Pf} (d_0, \cdots, \hat{a}_i, \cdots) \mathrm{Pf} (d_0, \cdots, \hat{a%
}_j, \cdots)\,.
\end{eqnarray*}
Therefore,
\begin{eqnarray}
&& - \frac{\partial f} {\partial s} \frac{\partial f} {\partial s}
=\sum_{i,j=1}^{2N} (-1)^{i+j+1} p_i^{-2} \mathrm{Pf} (a_i, a_j) \mathrm{Pf}
(d_0, \cdots, \hat{a}_i, \cdots) \mathrm{Pf} (d_0, \cdots, \hat{a}_j, \cdots)
\nonumber \\
&&=\sum_{i=1}^{2N} (-1)^{i+1} p_i^{-2} \mathrm{Pf} (d_0, \cdots, \hat{a}_i,
\cdots) \left[ \sum_{j=1}^{2N} (-1)^{j} \mathrm{Pf} (a_i, a_j) \mathrm{Pf}
(d_0, \cdots, \hat{a}_j, \cdots) \right]  \nonumber  \label{CSP1_proof5}
\end{eqnarray}
Further, we note that the following identity can be substituted into the
term within bracket
\begin{eqnarray*}
&& \sum_{j=1}^{2N} (-1)^{j} \mathrm{Pf} (a_i, a_j) \mathrm{Pf} (d_0, \cdots,
\hat{a}_j, \cdots) \\
&& = \mathrm{Pf} (d_{0}, a_i) \mathrm{Pf} (\cdots) + (-1)^{i+1} \mathrm{Pf}
(d_0, \cdots, \hat{b}_i, \cdots)\,
\end{eqnarray*}
which is obtained from the expansion of the following vanishing pfaffian $%
\mathrm{Pf} (a_i, d_0, \cdots)$ on $a_i$. Consequently, we have
\begin{eqnarray}
&& - \frac{\partial f} {\partial s} \frac{\partial f} {\partial s} =
\nonumber \\
&& \sum_{i=1}^{2N} (-1)^{i+1} p_i^{-2} \mathrm{Pf} (d_0, \cdots, \hat{a}_i,
\cdots) \left[\mathrm{Pf} (d_{0}, a_i) \mathrm{Pf} (\cdots) + (-1)^{i+1}
\mathrm{Pf} (d_0, \cdots, \hat{b}_i, \cdots)\right]\,,  \nonumber \\
&& = -\mathrm{Pf} (\cdots) \mathrm{Pf} (d_{-2}, d_0, \cdots)+
\sum_{i=1}^{2N} p_i^{-2} \mathrm{Pf} (d_0, \cdots, \hat{a}_i, \cdots)
\mathrm{Pf} (d_0, \cdots, \hat{b}_i, \cdots)\,,  \nonumber \\
&& = -\frac{\partial^2 f} {\partial s^2} f+ \frac 14 \sum_{1 \le \mu < \nu
\le n} c_{\mu \nu} g_{\mu} g_{\nu} \,,
\end{eqnarray}
which can be rewritten as
\begin{equation}
2\frac{\partial^2 f} {\partial s^2} f- 2\frac{\partial f} {\partial s} \frac{%
\partial f} {\partial s}= \frac 12 \sum_{1 \le \mu < \nu \le n} c_{\mu \nu}
g_{\mu} g_{\nu} \,.
\end{equation}
The above equation is nothing but the second bilinear equation. Thus, the
proof is complete.
\end{proof}
\section{Semi-discrete analogue of the multi-component short pulse equation}
In this section, we attempt to construct an integrable semi-discretization of the
MCSP equation (\ref{CSPE}). Firstly, we propose a
semi-discrete analogue of bilinear equations (\ref{CSP_Bilinear})
\begin{equation}
\left\{
\begin{array}{l}
\displaystyle\frac{1}{a}D_{s}(g_{k+1}^{(j)}\cdot f_{k}-g_{k}^{(j)}\cdot
f_{k+1})=g_{k+1}^{(j)}f_{k}+g_{k}^{(j)}f_{k+1}\,,\quad j=1,2,\cdots ,n\,, \\%
[5pt]
\displaystyle D_{s}^{2}f_{k}\cdot f_{k}=\frac{1}{2}\sum_{1\leq i<j\leq
n}c_{ij}g_{k}^{(i)}g_{k}^{(j)}\,.%
\end{array}%
\right.
\label{eq:sdBL1}
\end{equation}
%which can also be viewed as a B\"acklund transform of the bilinear equations (\ref{CSP_Bilinear}).
By introducing a dependent variable transformation
\begin{equation}
u_{k}^{(i)}=\frac{g_{k}^{(i)}}{f_{k}},\quad i=1,2,\cdots ,n\,,
\label{eq:v-trf}
\end{equation}%
and a discrete version of the hodograph transformation
\begin{equation}
x_{k}=2ka-2(\ln f_{k})_{s},\quad t=s\,,
\label{eq:sd-hodtrf}
\end{equation}%
the second bilinear equation in (\ref{eq:sdBL1}) is rewritten as
\begin{equation}
\left( \ln f_{k}\right) _{ss}=\frac{1}{4}\sum_{1\leq i<j\leq
n}c_{ij} \frac{g_{k}^{(i)}g_{k}^{(j)}}{f_{k}^{2}}=\frac{1}{4}\sum_{1\leq i<j\leq
n}c_{ij}u_{k}^{(i)}u_{k}^{(j)}\,.  \label{sd_BL2}
\end{equation}%
From the discrete hodograph transformation, we can define an nonuniform mesh
\begin{equation}
\delta_k=x_{k+1}-x_{k}=2a-2\left( \ln \frac{f_{k+1}}{f_{k}}\right) _{s}\,,
\label{delta}
\end{equation}%
it then immediately follows
\begin{equation}
\frac{d\delta_k}{d\,s}=-\frac{1}{2}\sum_{1\leq i<j\leq n}c_{ij}\left(
u_{k+1}^{(i)}u_{k+1}^{(j)}-u_{k}^{(i)}u_{k}^{(j)}\right)   \label{sd1-CSP}
\end{equation}%
from Eq. (\ref{sd_BL2}). Next, dividing both sides by $f_{k+1}f_{k}$, the
first bilinear equation in (\ref{eq:sdBL1}) can be calculated out by
\[
\left( \frac{g_{k+1,s}^{(i)}}{f_{k+1}}-\frac{g_{k,s}^{(i)}}{f_{k}}\right) -%
\frac{g_{k+1}^{(i)}f_{k,s}-g_{k}^{(i)}f_{k+1,s}}{f_{k+1}f_{k}}=a\left( \frac{%
g_{k+1}^{(i)}}{f_{k+1}}+\frac{g_{k}^{(i)}}{f_{k}}\right) \,,
\]%
or
\[
\left( \frac{g_{k+1}^{(i)}}{f_{k+1}}-\frac{g_{k}^{(i)}}{f_{k}}\right)
_{s}+\left( \frac{g_{k+1}^{(i)}}{f_{k+1}}+\frac{g_{k}^{(i)}}{f_{k}}\right)
\left( \frac{f_{k+1,s}}{f_{k+1}}-\frac{f_{k,s}}{f_{k}}\right) =a\left( \frac{%
g_{k+1}^{(i)}}{f_{k+1}}+\frac{g_{k}^{(i)}}{f_{k}}\right) \,,
\]%
which is recast into
\begin{equation}
\left( \frac{g_{k+1}^{(i)}}{f_{k+1}}-\frac{g_{k}^{(i)}}{f_{k}}\right)
_{s}=\left( a-\left( \ln \frac{f_{k+1}}{f_{k}}\right) _{s}\right) \left(
\frac{g_{k+1}^{(i)}}{f_{k+1}}+\frac{g_{k}^{(i)}}{f_{k}}\right) \,.
\end{equation}%
With the use of Eqs. (\ref{eq:v-trf}) and (\ref{delta}), we finally arrive at
\begin{equation}
\frac{d(u_{k+1}^{(i)}-u_{k}^{(i)})}{ds}=\frac{1}{2}%
(x_{k+1}-x_{k})(u_{k+1}^{(i)}+u_{k}^{(i)})\,.  \label{sd2-CSP}
\end{equation}
Eqs. (\ref{sd1-CSP}), (\ref{sd2-CSP}) constitute the semi-discrete
analogue of the MCSP equation. We summarize the results by the following
Theorem.
\begin{theorem}
The bilinear equations (\ref{eq:sdBL1})
yield a semi-discrete multi-component short pulse equation
\begin{equation}
\left\{
\begin{array}{l}
\displaystyle \frac{d(u_{k+1}^{(j)}-u_{k}^{(j)})}{ds} =\frac{1}{2}%
\delta_k (u_{k+1}^{(j)}+u_{k}^{(j)})\,, \\
[5pt]
\displaystyle \frac{d \delta_k}{ds} =-\frac{1}{2}\sum_{1\leq i<j\leq n}c_{ij
}\left(u_{k+1}^{(i)}u_{k+1}^{(j)}-u_{k}^{(i)}u_{k}^{(j)}\right)\,.
\end{array}%
\right.
\label{eq:sd-MSP}
\end{equation}
through dependent variable transformation
\begin{equation*}
u_{k}^{(i)}=\frac{g_{k}^{(i)}}{f_{k}},\quad i=1,2,\cdots ,n\,,
\end{equation*}%
and discrete hodograph transformation
\begin{equation*}
x_{k}=2ka-2(\ln f_{k})_{s},\quad t=s\,,
\end{equation*}
where $\delta_k=x_{k+1}-x_k$.
\end{theorem}

To assure its integrability, we provide its multi-soliton solution in terms
of pfaffians by the following theorem. The elements of the pfaffians are defined as follows:

\begin{equation}  \label{sd-pf1}
\mathrm{Pf}(a_i,a_j)_k= \frac{p_i-p_j}{p_i+p_j}\varphi_i^{(0)}(k)%
\varphi_j^{(0)}(k)\,, \quad \mathrm{Pf}(a_i,b_j)_k=\delta_{i,j}\,,
\end{equation}

\begin{equation}  \label{sd-pf2}
\mathrm{Pf}(b_i,b_j)_k=\frac 14 \frac{c_{\mu \nu }}{p^{-2}_i-p^{-2}_{j}}\,,
(b_i \in B _\mu, b_j \in B_\nu)\,,
\end{equation}

\begin{equation}  \label{sd-pf3}
\mathrm{Pf}(d_l,a_i)_k= \varphi_i^{(l)}(k)\,, \quad \mathrm{Pf}(a_i, d^k)_k=
\varphi_i^{(0)}(k+1)\,,
\end{equation}

%! n -->0
\begin{equation}  \label{sd-pf4}
\mathrm{Pf}(b_j,\beta_\mu)_k= \left\{
\begin{array}{ll}
1 & \quad b_j \in B_\mu \\
0 & \quad b_j \notin B_\mu%
\end{array}%
\right.
\end{equation}

\begin{equation}  \label{sd-pf5}
\mathrm{Pf}(d_0, d^k)_k=1, \quad \mathrm{Pf}(d_{-1}, d^k)_k=-a\,,
\end{equation}
%\begin{equation} \label{sd-pf6}
%{\rm Pf}(d_l,\beta_\mu)_k= {\rm Pf}(d_l,b_j)_k= {\rm Pf}(a_i,\beta_\mu)_k =0\,,
%\end{equation}
where
\[
\varphi_i^{(n)}(k) =p_i^n\left(\frac{1+ap_i}{1-ap_i}\right)^ke^{\xi_i},
\quad \xi_i=p_i^{-1} s +\xi_{i0}\,.
\]
Here $i,j=1, 2, \cdots, 2N$, $\mu, \nu= 1, 2, \cdots, n$ and $k,l$ are
arbitrary integers. Other pfaffian elements not mentioned above are all
zeros. Note that $\varphi_i^{(n)}(k)$ has the following property
\[
\frac{\varphi_i^{(n)}(k+1)-\varphi_i^{(n)}(k)}{a} =
\varphi_i^{(n+1)}(k+1)+\varphi_i^{(n+1)}(k)\,,
\]
which is used in the proof of the theorem.

\begin{theorem}
The bilinear equations (\ref{eq:sdBL1})
have the following pfaffian solution
\begin{eqnarray}
f_{k} &=&\mathrm{Pf}(a_{1},\cdots ,a_{2N},b_{1},\cdots ,b_{2N})_{k}\,,
\label{sd_solution} \\
g_{k}^{(j)} &=&\mathrm{Pf}(d_{0}, \beta _{j},a_{1},\cdots
,a_{2N},b_{1},\cdots ,b_{2N})_{k}\,,
\end{eqnarray}%
where $j=1,2,\cdots ,n$, the elements of pfaffians are defined in eqs. (\ref%
{sd-pf1})--(\ref{sd-pf5}).
\end{theorem}

\begin{proof}
Since
\[
\frac{\partial}{\partial s} \mathrm{Pf} (a_i,a_j)_{k}
=\varphi_i^{(0)}(k)\varphi_j^{(-1)}(k)
-\varphi_i^{(-1)}(k)\varphi_j^{(0)}(k)=\mathrm{Pf} (d_{-1}, d_0,
a_i,a_j)_{k}\,,
\]

\begin{eqnarray*}
\mathrm{Pf}(a_i,a_j)_{k+1}&=& \mathrm{Pf}(a_i,a_j)_{k}
+\varphi_i^{(0)}(k+1)\varphi_j^{(0)}(k)
-\varphi_i^{(0)}(k)\varphi_j^{(0)}(k+1) \\
&=& \mathrm{Pf} (d_{0}, d^k, a_i,a_j)_{k}\,,
\end{eqnarray*}

\begin{eqnarray*}
&& \left(\partial_s-a \right) \mathrm{Pf}(a_i,a_j)_{k+1} \\
&& \quad =\varphi_i^{(0)}(k+1)\varphi_j^{(-1)}(k+1)
-\varphi_i^{(-1)}(k+1)\varphi_j^{(0)}(k+1) \\
&& \qquad -a \left( \mathrm{Pf} (a_i,a_j)_{k} +
\varphi_i^{(0)}(k+1)\varphi_j^{(0)}(k)
-\varphi_i^{(0)}(k)\varphi_j^{(0)}(k+1)\right) \\
&& \quad =-a \mathrm{Pf} (a_i,a_j)_{k} + \varphi_i^{(0)}(k+1) \left(
\varphi_j^{(-1)}(k+1)-a \varphi_j^{(0)}(k) \right) \\
&& \qquad - \varphi_j^{(0)}(k+1) \left( \varphi_i^{(-1)}(k+1) -a
\varphi_i^{(0)}(k) \right) \\
&& \quad = -a \mathrm{Pf} (a_i,a_j)_{k} + \varphi_i^{(0)}(k+1) \left( a
\varphi_j^{(0)}(k+1)+ \varphi_j^{(-1)}(k) \right) \\
&& \qquad - \varphi_j^{(0)}(k+1) \left( a\varphi_i^{(0)}(k+1) +
\varphi_i^{(-1)}(k) \right) \\
&& \quad =-a \mathrm{Pf} (a_i,a_j)_{k} + \varphi_i^{(0)}(k+1)
\varphi_j^{(-1)}(k) - \varphi_i^{(-1)}(k) \varphi_j^{(0)}(k+1) \\
&& \quad = \mathrm{Pf} (d_{-1}, d^k, a_i,a_j)_{k}\,,
\end{eqnarray*}
we have
\[
{\partial_s} f_k = \mathrm{Pf} (d_{-1}, d_0, \cdots)_k\,,
\]

\[
f_{k+1} = \mathrm{Pf} (d_0,d^k \cdots)_k\,,
\]
\[
\left(\partial_s-a \right) f_{k+1} = \mathrm{Pf} (d_{-1},d^k \cdots)_k\,.
\]
Furthermore, we can verify
\begin{eqnarray*}
&& {\partial_s} g^{(\mu)}_k = \partial_s \left(\sum_{i=1}^{2N} (-1)^{i}
\mathrm{Pf} (d_0, a_i)_k \mathrm{Pf} (\beta_\mu, \cdots ,\hat{a}_i,
\cdots)_k\right) \\
&& = \sum_{i=1}^{2N} (-1)^{i} \left(\left(\partial_s \mathrm{Pf} (d_0,
a_i)_k \right) \mathrm{Pf} (\beta_\mu, \cdots ,\hat{a}_i, \cdots)_k +
\mathrm{Pf} (d_0, a_i)_k {\partial_s} \mathrm{Pf} (\beta_\mu, \cdots ,\hat{a}%
_i, \cdots)_k \right) \\
&& =\sum_{i=1}^{2N} (-1)^{i} \left( \mathrm{Pf} (d_{-1}, a_i)_k \mathrm{Pf}
(\beta_\mu, \cdots ,\hat{a}_i, \cdots)_k + \mathrm{Pf} (d_0, a_i)_k \mathrm{%
Pf} (\beta_\mu, d_{-1}, d_0, \cdots ,\hat{a}_i, \cdots)_k \right) \\
&& = \mathrm{Pf} (d_{-1}, \beta_\mu,  \cdots)_k + \mathrm{Pf} (d_0, \beta_\mu,
 d_{-1}, d_0, \cdots)_k \\
&& = \mathrm{Pf} (d_1, \beta_\mu,  \cdots)_k\,,
\end{eqnarray*}

\begin{eqnarray*}
&& g^{(\mu)}_{k+1} = \sum_{i=1}^{2N}(-1)^{i}\mathrm{Pf}(d_0,a_i)_{k+1}
\mathrm{Pf}(\beta_\mu,\cdots,\hat{a}_i,\cdots)_{k+1} \\
&& = \sum_{i=1}^{2N}(-1)^{i-1}\mathrm{Pf}(d^k,a_i)_k \mathrm{Pf}%
(\beta_\mu,d_0,d^k,\cdots,\hat{a}_i,\cdots)_k \\
&& = \sum_{i=1}^{2N}(-1)^{i-1}\mathrm{Pf}(d^k,a_i)_k \left( \mathrm{Pf}%
(\beta_\mu,\cdots,\hat{a}_i,\cdots)_k +\sum_{j=1}^{i-1}(-1)^j\mathrm{Pf}
(\beta_\mu,d^k,\cdots,\hat{a}_j,\cdots,\hat{a}_i,\cdots)_k \right. \\
&& \left. +\sum_{j=i+1}^{2N}(-1)^{j-1}\mathrm{Pf} (\beta_\mu,d^k,\cdots,\hat{%
a}_i,\cdots,\hat{a}_j,\cdots)_k\right) \\
&& =\mathrm{Pf} (d^k,  \beta_\mu, \cdots)_k\,,
\end{eqnarray*}

\begin{eqnarray*}
&& ({\partial_s}-a) g^{(\mu)}_{k+1} = \sum_{i=1}^{2N}(-1)^{i-1} \left(({%
\partial_s}-a)\mathrm{Pf}(d^k,a_i)_k\right) \mathrm{Pf}(\beta_\mu,\cdots,%
\hat{a}_i,\cdots)_k \\
&& +\sum_{1\le i<j\le 2N}(-1)^{i+j-1}\left(\partial_s\mathrm{Pf}(a_i,a_j)_k
\right)\mathrm{Pf}(\beta_\mu,d^k,\cdots,\hat{a}_i,\cdots,\hat{a}_j,\cdots)_k
\\
&& = \sum_{i=1}^{2N}(-1)^{i-1} \left(\mathrm{Pf}(d_{-1},a_i)_k+a\mathrm{Pf}%
(d_0,a_i)_k\right) \mathrm{Pf}(\beta_\mu,\cdots,\hat{a}_i,\cdots)_k \\
&& +\sum_{1\le i<j\le 2N}(-1)^{i+j-1}\mathrm{Pf}(d_{-1},d_0,a_i,a_j)_k
\mathrm{Pf}(\beta_\mu,d^k,\cdots,\hat{a}_i,\cdots,\hat{a}_j,\cdots)_k \\
&& =\mathrm{Pf} (d_{-1}, \beta_\mu, d_0, d^k, \cdots)_k\,.
\end{eqnarray*}

Therefore, an algebraic identity of pfaffian
\begin{eqnarray*}
&&\mathrm{Pf}(d_{-1}, \beta _{\mu },d_{0},d^{k},\cdots )_{k}\mathrm{Pf}%
(\cdots )_{k}=\mathrm{Pf}(d_{-1},d_{0},\cdots )_{k}\mathrm{Pf}(d^{k}, \beta _{\mu
},\cdots )_{k} \\
&&\quad -\mathrm{Pf}(d_{-1},d^{k},\cdots )_{k}\mathrm{Pf}(d_{0}, \beta _{\mu
},\cdots )_{k}+\mathrm{Pf}(d_{-1}, \beta _{\mu },\cdots )_{k}\mathrm{Pf}%
(d_{0},d^{k},\cdots )k\,,
\end{eqnarray*}%
together with above pfaffian relations gives
\[
({\partial _{s}}-a)g_{k+1}^{(\mu )}\times f_{k}=g_{k+1}^{(\mu )}\times {%
\partial _{s}}f_{k}-({\partial _{s}}-a)f_{k+1}\times g_{k+1}^{(\mu )}+{%
\partial _{s}}g_{k}^{(\mu )}\times f_{k+1}\,,
\]%
which is nothing but the first bilinear equation. The second bilinear
equation can be proved in a similar way as in the continuous case.
\end{proof}

\begin{remark}
Bilinear equations (\ref{eq:sdBL1}) can be viewed as a B\"{a}cklund
transformation of bilinear equations (\ref{CSP_Bilinear}), which yield
the MCSP equation. In other words, if $f_{k}$ and $%
g_{k}^{(j)}$ satisfy (\ref{CSP_Bilinear}), so do $f_{k+1}$ and $g_{k+1}^{(j)}
$ . Based on (\ref{eq:sdBL1}), we propose a semi-discrete analogue of
the MCSP equation. The integrability of the semi-discrete MCSP equation is guaranteed by the existence of $N$-soliton solution.
\end{remark}

Finally, let us show that in the continuous limit, $a\to 0$ ($\delta_k\to 0$),
the proposed semi-discrete multi-component short pulse equation recovers the continuous one (%
\ref{CPE1}). The dependent variable $u$ is regarded as a function of $x$ and
$t$, where $x$ is the space coordinate of the $k$-th lattice point and $t$
is the time, defined by
%! I added two for the first line, and I found a typo in the third one.
\[
x=x_0+\sum_{j=0}^{k-1}\delta_j\,,\qquad t=s\,,
\]
where $\delta_j=x_{j+1}-x_j$. In the continuous limit, $a\to 0$ ($%
\delta_k\to 0$), we have
\[
\frac 12 (u^{(i)}_{k+1}+u^{(i)}_k) \to u_{i}\,, \quad \frac{%
\partial_s(u^{(i)}_{k+1}-u^{(i)}_k)}{\delta_k} \to u_{i,sx}\,,
\]

\begin{eqnarray*}
&& \frac{\partial x}{\partial s} =\frac{\partial x_0}{\partial s}
+\sum_{j=0}^{k-1}\frac{\partial\delta_j}{\partial s} \\
&& \qquad =\frac{\partial x_0}{\partial s}-\frac 12 \sum_{1\le \mu< \nu \le
n} c_{\mu \nu} \sum_{j=0}^{k-1}\left( u^{(\mu)}_{j+1}u^{(\nu)}_{j+1}-
u^{(\mu)}_{j}u^{(\nu)}_{j}\right) \\
&& \qquad \to -\frac 12 \sum_{1\le \mu< \nu \le n} c_{\mu \nu}u_{\mu}u_{\nu}
\,,
\end{eqnarray*}
hence
\[
\partial_s=\partial_t+\frac{\partial x}{\partial s}\partial_x \to\partial_t
-\frac 12 \left(\sum_{1\le \mu< \nu \le n} c_{\mu \nu}u_{\mu}u_{\nu}\right)
\partial_x\,,
\]
where the origin of space coordinate $x_0$ is taken so that $\displaystyle%
\frac{\partial x_0}{\partial s}$ cancels $\frac 12 \sum_{1\le \mu< \nu \le
n} c_{\mu \nu}u^{(\mu)}_0u^{(\nu)}_0$. Thus the first semi-discrete
multi-component SP equation converges to
\[
\partial_x\left(\partial_t-\frac 12 \left(\sum_{1\le \mu< \nu \le n} c_{\mu
\nu}u_{\mu}u_{\nu}\right) \partial_x\right)u_{i} =u_{i}\,,
\]
which is exactly the MCSP equation (\ref{CSPE}).
\section{Two-component short pulse equation}
Since the two-component short pulse equation is of particular importance for applications in nonlinear optics,
we provide a detailed study for this two-component system, together with its
semi-discrete analogue in this section. For the continuous case of $n=2$, we can take $c_{12}=1$ without loss of generality and arrive at the following two-component system \cite{Matsuno_CSPE}
\begin{eqnarray}
u_{xt} & = & u+ \frac 12 \left( uv u_x \right)_{x}\,,  \label{2CSPEa1} \\
v_{xt} & = & v+ \frac 12 \left( uv v_x \right)_{x}\,,  \label{2CSPEa2}
\end{eqnarray}
where $u=u_{1}, v=u_{2}$.  Furthermore, if we assume $u$ is a complex-values function and impose a complex conjugate condition $v=\bar{u}$,
where $\bar{u}$ means the complex conjugate of $u$. Eq.
 (\ref{2CSPEa1}) leads to a \textbf{complex short pulse equation} studied in \cite{comSPE,dcomSPE}
\begin{equation}  \label{complex_SP}
u_{xt} = u+ \frac 12 \left( |u|^2 u_x \right)_{x}\,.
\end{equation}
% The same notation is used for other variables.
Since the complex short pulse equation (\ref{complex_SP}) is a special case of
two-component system (\ref{2CSPEa1})--(\ref{2CSPEa2}), its $N$-soliton
solution can be obtained from the $N$-soliton solution of system (\ref%
{2CSPEa1})--(\ref{2CSPEa2}) by requiring $f={\bar {f}}$, $g^{(1)}={\bar g}^{(2)}$.
These requirements can be achieved by putting $p_{k}={\bar p}_j$ and
$\xi_{k0}={\bar \xi}_{j0}$, where $k=j+N$, $j=1, 2, \cdots, N$.

In particular, the tau-functions for one-soliton solution  ($N=1$) are found to be
\begin{eqnarray}
&& f =
-1-\frac 14 \frac {(p_1\bar{p}_1)^2}{(p_1+\bar{p}_1)^2}
e^{\eta_1+\bar{\eta}_1} \,,
\end{eqnarray}
\begin{equation}
%g = \mathrm{Pf} (d_0,\beta_1, a_1,a_2,b_1,b_2) = -\alpha_1 e^{\eta_1}\,.
g =  - e^{\eta_1}\,.
\end{equation}
Let $p_1=p_{1R}+ \mathrm{{i}} p_{1I}$, and we assume $p_{1R} >0$ without loss of
generality, then the one-soliton solution can be expressed in the following
parametric form
\begin{equation}  \label{CSP1solitona}
u=\frac{2p_{1R}}{|p_1|^2}e^{\mathrm{{i}} \eta_{1I}} %
\mbox{sech} \left(\eta_{1R}+\eta_{10}\right) \,,
\end{equation}
\begin{equation}  \label{CSP1solitonb}
x=y-\frac{2p_{1R}}{|p_1|^2}\left(\tanh \left(\eta_{1R}+\eta_{10}\right)+1\right)\,,
\quad t=s\,,
\end{equation}
where
\begin{equation}
\eta_{1R}=p_{1R}y+\frac{p_{1R}}{|p_1|^2}s , \quad \eta_{1I}=p_{1I} y-\frac{%
p_{1I}}{|p_1|^2}s \,,\quad \eta_{10}=\ln \frac{|p_1|^2}{4p_{1R}}\,.
\end{equation}

Eq. (\ref{CSP1solitona}) represents an envelope soliton of amplitude $%
2p_{1R}/|p_1|^2$ and phase $\eta_{1I}$. The details analysis concerning its property was carried out
in \cite{comSPE}. In summary，three types can be classified.
\begin{itemize}
\item when $|p_{1R}| < |p_{1I}|$, it is a \textbf{smooth soliton} solution, which is
similar to the envelope soliton solution for the nonlinear Schr{\"o}dinger equation.
\item when $|p_{1R}| > |p_{1I}|$, it is a \textbf{loop soliton} solution, which admits multi-valued property.
\item  when $|p_{1R}| = |p_{1I}|$, it is a \textbf{cuspon soliton}, this is a case which divides the single-valued and
multi-valued solution.
\end{itemize}
\subsection{Semi-discrete two-component system}
Based on the results in previous section, we have an integrable semi-discrete analogue of
two-component system (\ref{2CSPEa1})--(\ref{2CSPEa2})
\begin{equation}
\left\{
\begin{array}{l}
\displaystyle \frac{d (u_{k+1}-u_k)} {dt}  = \frac 12 \delta_k (u_{k+1}+u_k), \\
\displaystyle \frac{d (v_{k+1}-v_k)} {dt}  = \frac 12 \delta_k (v_{k+1}+v_k), \\
\displaystyle \frac{d \delta_k} {dt}  =  -\frac 12 (u_{k+1}v_{k+1}-u_kv_k),
\end{array}%
\right.   \label{sd_2CSP}
\end{equation}%
which admit the following $N$-soliton solution in parametric form
\begin{equation}
u_k=\frac{g^{(1)}_{k}}{f_k}, \quad v_k=\frac{g^{(2)}_{k}}{f_k}
\end{equation}
and hodograph transformation
\begin{equation}
x_k = 2ka -2(\ln f_k)_s, \quad t=s\,,
\end{equation}
with
\begin{equation}
f_k= \mathrm{Pf} (a_1, \cdots, a_{2N}, b_1, \cdots, b_N, c_1, \cdots,
c_N)_k\,,
\label{eq:twocomNsol1}
\end{equation}
\begin{equation}
g^{(i)}_k= \mathrm{Pf} (\beta_i, d_0, a_1, \cdots, a_{2N}, b_1, \cdots, b_N,
c_1, \cdots, c_N)_k\,, \quad i=1,2
\label{eq:twocomNsol2}
\end{equation}
where the elements of pfaffians are (others not mentioned are all zeros)
\[
\mathrm{Pf}(a_i,a_j)_k= \frac{p_i-p_j}{p_i+p_j}\varphi_i^{(0)}(k)%
\varphi_j^{(0)}(k)\,, \quad \mathrm{Pf}(a_i,b_j)_k=\delta_{i,j},
\]

\[
\mathrm{Pf} (b_i,c_j)=-\frac 14 \frac{(p_ip_{N+j})^2}{p^2_i-p^2_{N+j}} \,,
\quad \mathrm{Pf} (a_i,c_j)_k = \delta_{i,j+N}\,,
\]

\[
\mathrm{Pf}(b_i,\beta_1)=\mathrm{Pf}(c_i,\beta_2)=1\,, \quad \mathrm{Pf}%
(d_0,a_i)_k= \varphi_i^{(0)}(k) \,.
\]
By imposing complex conjugate conditions $p_{k}={\bar p}_j$, $\xi_{k0}={\bar \xi}_{j0}$
($k=j+N$, $j=1,2, \cdots, N$), it then follows $f_k={\bar f}_k$, $%
g^{(1)}_k={\bar g}^{(2)}_k$, thus $u_k={\bar v}_k$, which leads to a
semi-discrete analogue of the complex short equation (\ref{complex_SP})
\begin{equation}
\left\{
\begin{array}{l}
\displaystyle \frac{d (u_{k+1}-u_k)} {dt}  = \frac 12 \delta_k (u_{k+1}+u_k)\,, \\
\displaystyle \frac{d \delta_k} {dt}  =  -\frac 12 (|u_{k+1}|^2-|u_k|^2)\,. \\
\end{array}%
\right.   \label{sd_2ComSP}
\end{equation}%
Its $N$-soliton solution immediately follows from the $N$-soliton solution
of Eqs. (\ref{eq:twocomNsol1})--(\ref{eq:twocomNsol2}) under complex conjugate
conditions mentioned above. In what follows, we list the one- and two-soliton solutions.

{\textbf{One-soliton solution:}}
The tau-functions for one-soliton solution to Eq. (\ref{sd_2ComSP}) are
\begin{equation}
f_k=-1-\frac 14 \frac {(p_1 p_2)^2}{(p_1+p_2)^2} \varphi_{12}^{(0)}(k)\,,
\quad g^{(1)}_k=-\varphi_1^{(0)}(k),
\end{equation}
with
\[
\varphi_i^{(0)}(k)=\left(\frac{1+ap_i}{1-ap_i}\right)^k
e^{p^{-1}_is+\xi_{i0}}, \quad
\varphi_{ij}^{(0)}(k)=\varphi_i^{(0)}(k)\varphi_j^{(0)}(k), \quad i,j=1,2\,,
\]
where $p_1={\bar p}_2$, $\xi_{10}={\bar \xi}_{20}$.
Similar to the continuous case,  if $p_1=p_{1R}+ \mathrm{{i}}p_{1I}$, we then arrive at the one-soliton solution of
semi-discrete complex short pulse equation (\ref{sd_2ComSP})

\begin{equation}  \label{sd_2Com_SP1solitona}
u_k=\frac{2p_{1R}}{|p_1|^2}  e^{\mathrm{{i}}
\chi_k}\mbox{sech} (\theta_k+\theta_0) \,,
\end{equation}
\begin{equation} \label{sd_2Com_SP1solitonb}
 x_k=2ka-\frac{2p_{1R}}{|p_1|^2}\left(\tanh
(\theta_k+\theta_0)+1\right)\,,
\end{equation}
where
\begin{equation}
 \theta_k=kd_1+\frac{p_{1R}}{|p_1|^2}s,
\quad \chi_k=kd_2-\frac{p_{1R}}{|p_1|^2}s\,, \quad \frac{1+ap_1}{1-ap_1}=e^{d_1+\mathrm{{i}}d_2}\,.
\end{equation}
In Fig. 1 (a)--(c), we illustrate the envelope soliton for $p_1=1+ 1.5\mathrm{{i}}$, $1+ \mathrm{{i}}$, $1+ 0.5\mathrm{{i}}$, which correspond to the smooth, cuspon and loop solition, respectively in the continuous case.
\begin{figure}[htbp]
\centerline{
\includegraphics[scale=0.38]{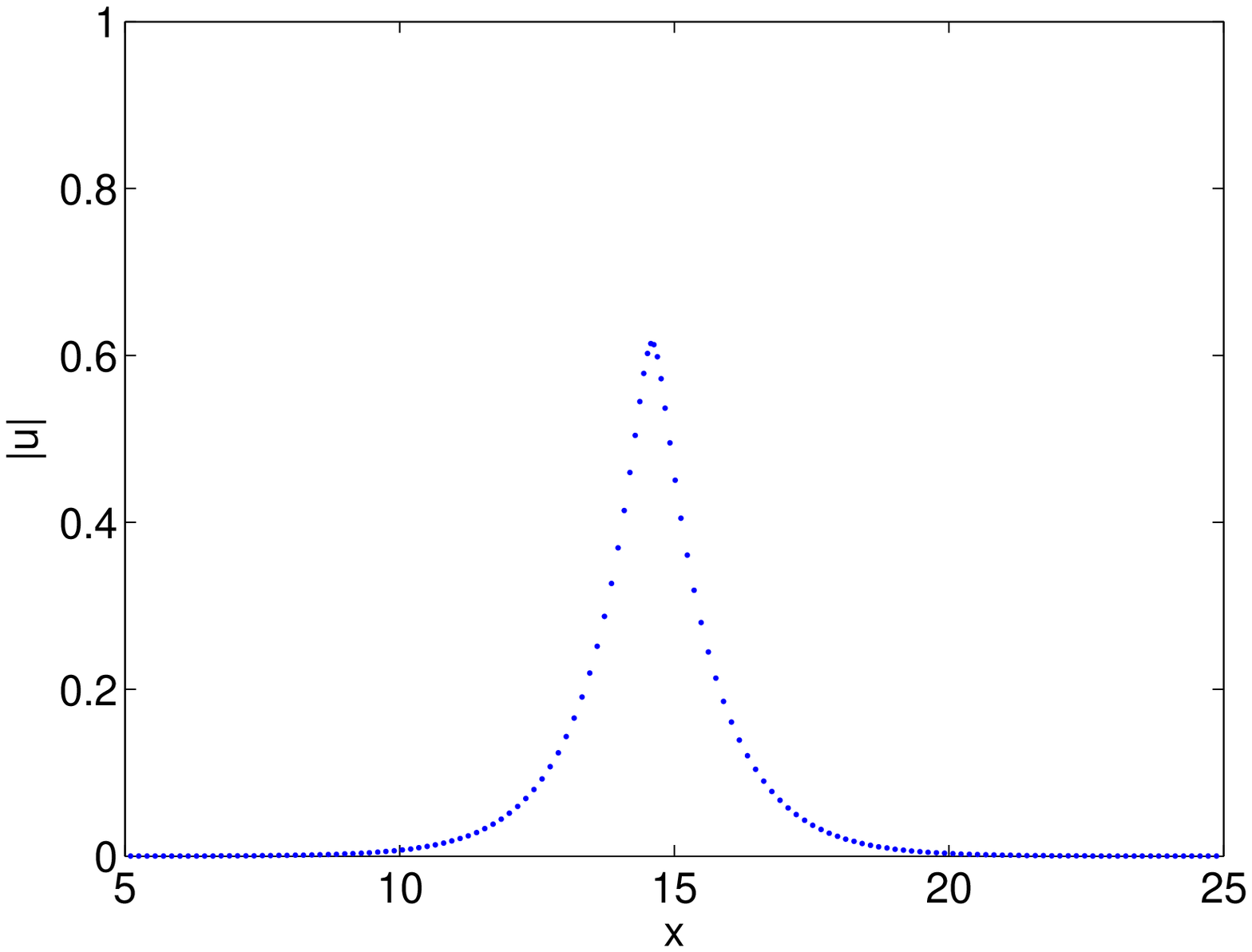}\quad
\includegraphics[scale=0.38]{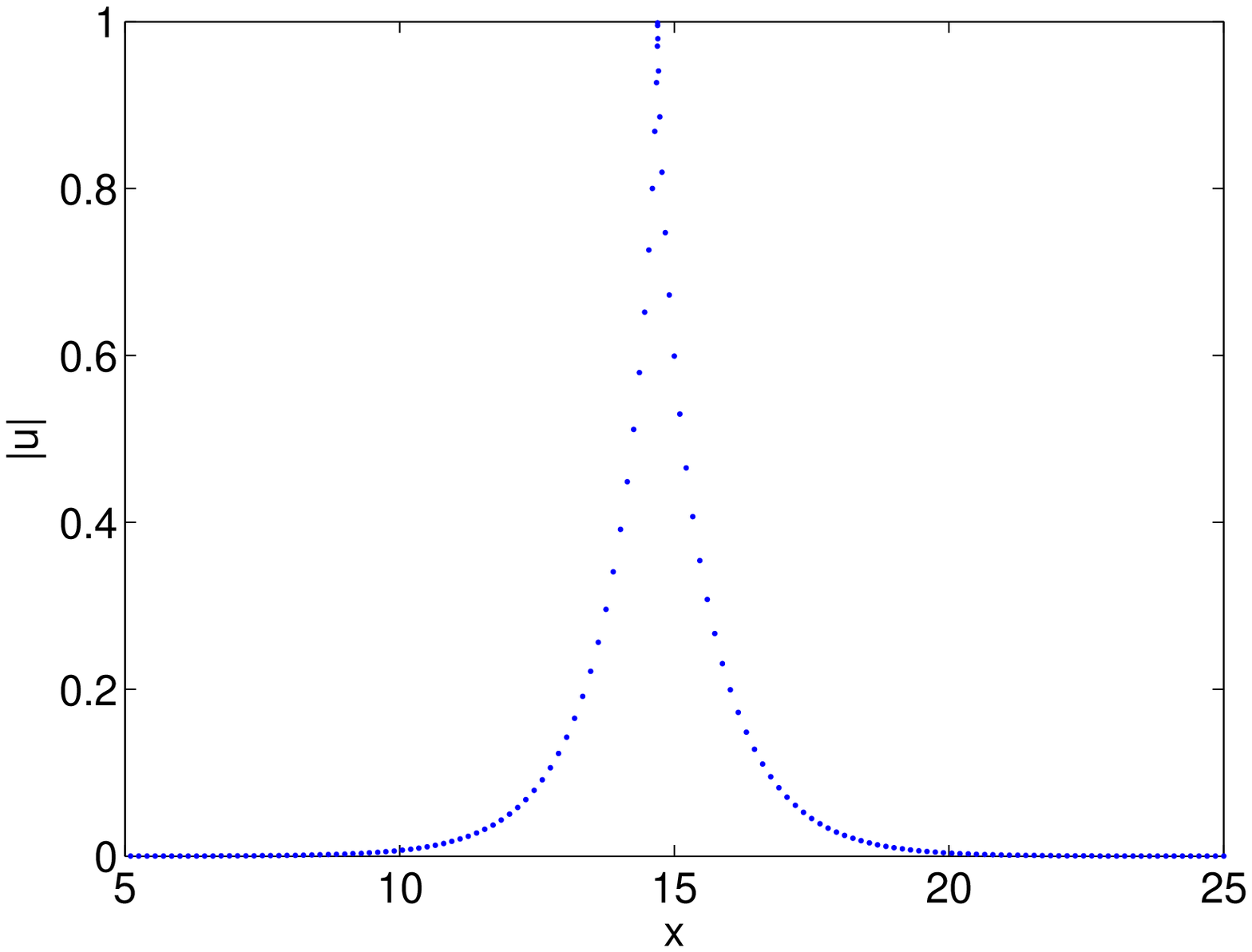}} \kern-0.3\textwidth
\hbox to
\textwidth{\hss(a)\kern13.5em\hss(b)\kern14em} \kern+0.3\textwidth
\centerline{
\includegraphics[scale=0.38]{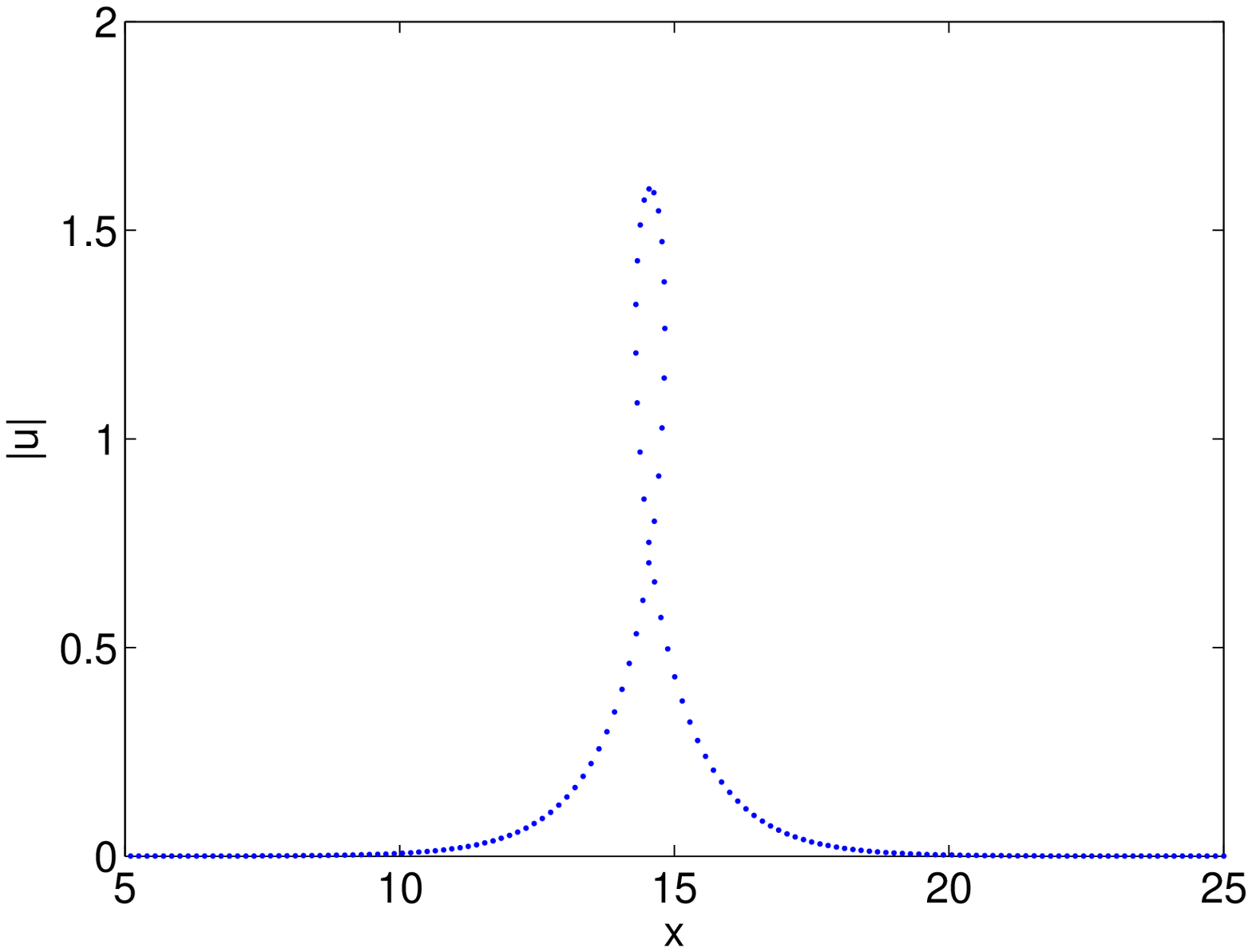}} \kern-0.3\textwidth
\hbox to
\textwidth{\hss(c)\kern23em} \kern+0.3\textwidth
\caption{Envelope soliton for the semi-discrete complex short pulse equation (a) smooth soliton
with $p_1=1+1.5\mathrm{{i}}$, (b) cuspon soliton with $p_1=1+\mathrm{{i}}$, (c) loop soliotn with $p_1=1+0.5\mathrm{{i}}$.}
\label{figure:cspe1soliton}
\end{figure}
{\textbf{Two-soliton:}} The tau-functions for two-soliton solutions to semi-discrete complex short
pulse equation (\ref{sd_2ComSP}) can be obtained by
\begin{eqnarray*}
&& f_k=1+a_{13}\varphi_{13}^{(0)}(k) + a_{14} \varphi_{14}^{(0)}(k) +a_{23}
\varphi_{23}^{(0)}(k) +a_{24} \varphi_{24}^{(0)}(k)  \nonumber \\
&& \qquad +16a_{12}a_{23}a_{14}a_{24}\left(p_1^{-1}-p_2^{-1} \right)^2
\left(p_3^{-1}-p_4^{-1} \right)^2
\varphi_{12}^{(0)}(k)(k)\varphi_{34}^{(0)}(k) \,,
\end{eqnarray*}

\begin{equation*}
g^{(1)}_k= \varphi_1^{(0)}(k)+ \varphi_2^{(0)}(k)+ 4\left( a_{13}a_{23}
\varphi_3^{(0)}(k)  + a_{14}a_{24} \varphi_4^{(0)}(k) \right)
\left(p_1^{-1}-p_2^{-1} \right)^2 \varphi_{12}^{(0)}(k) \,.
\end{equation*}
where
\begin{equation*}
a_{ij}=\frac{p_i-p_j}{p_i+p_j}\,, \quad p_1={\bar p}_3, \quad p_2={\bar p}_4
\end{equation*}
and $\xi_{10}={\bar \xi}_{30}$ and $\xi_{20}={\bar \xi}_{40}$. In Fig. %
2 (a)--(d), we show the process of interaction between a smooth envelop
soliton and a cuspon envelop soliton with $p_1=1+1.5\mathrm{{i}}$, $p_2=1+\mathrm{{i}}$, $%
\xi_{10}=-15$, $\xi_{20}=-25$.

\begin{figure}[htbp]
\centerline{
\includegraphics[scale=0.35]{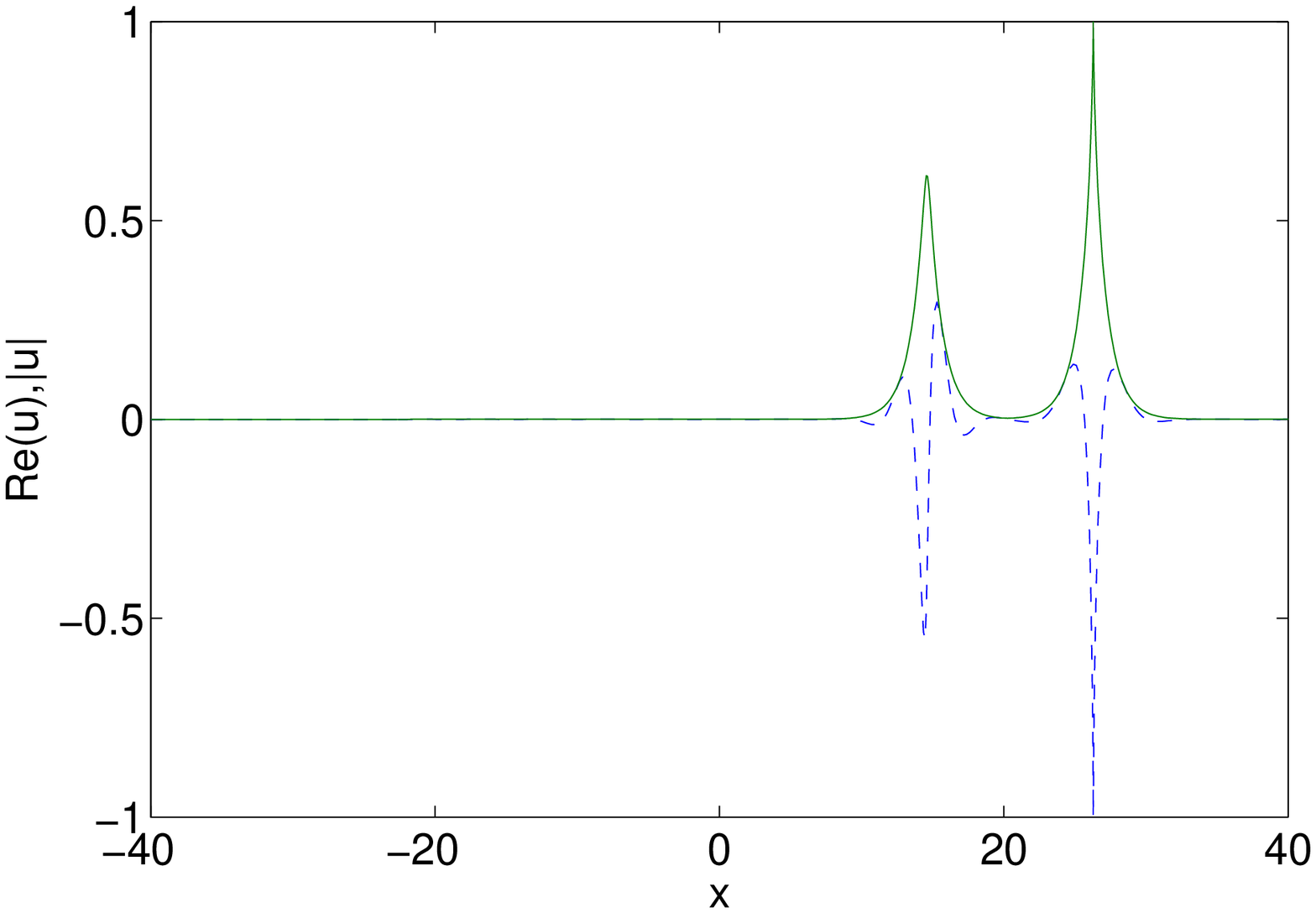}\quad
\includegraphics[scale=0.35]{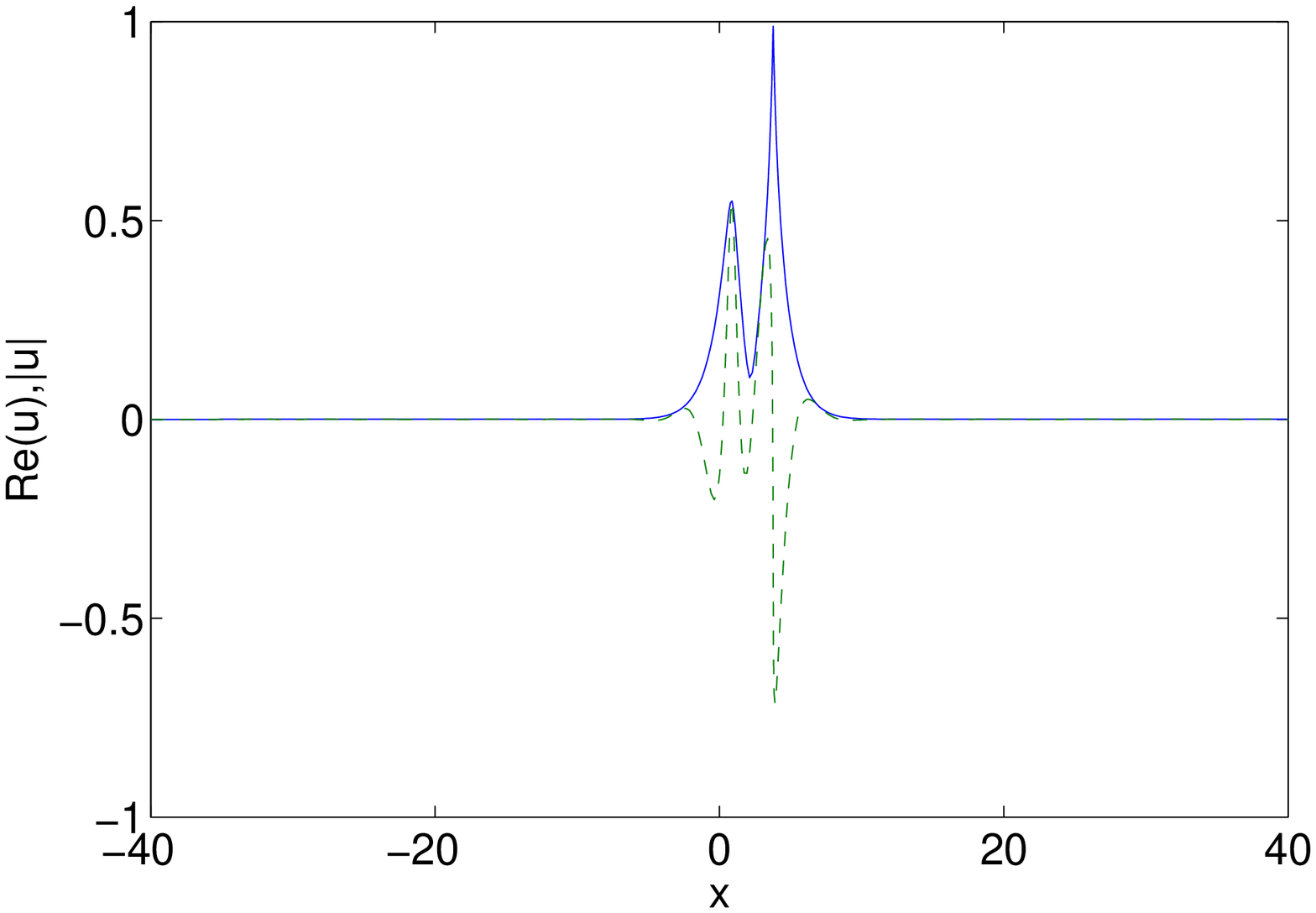}} \kern-0.3\textwidth
\hbox to
\textwidth{\hss(a)\kern0em\hss(b)\kern5em} \kern+0.25\textwidth
\centerline{
\includegraphics[scale=0.35]{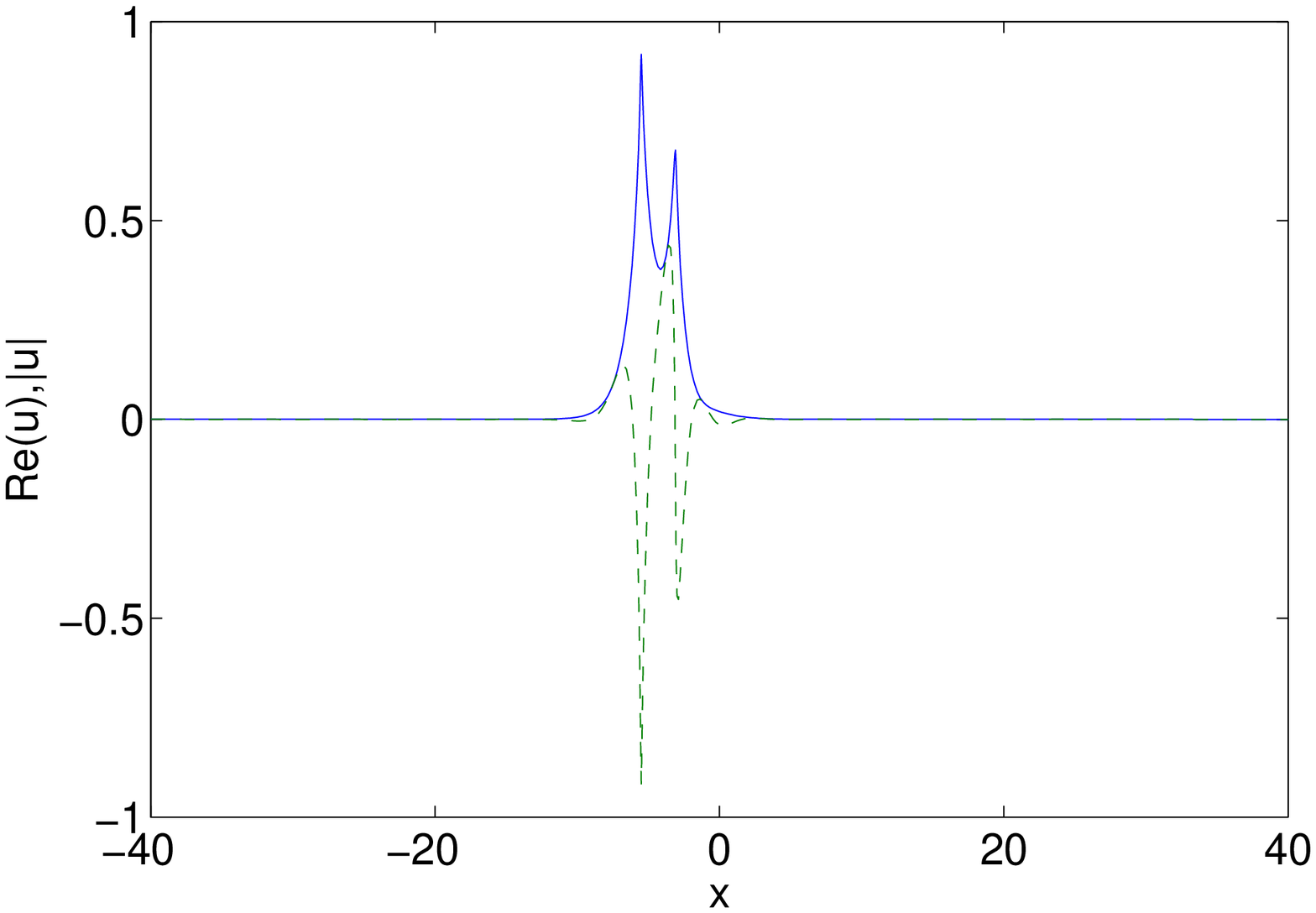}\quad
\includegraphics[scale=0.35]{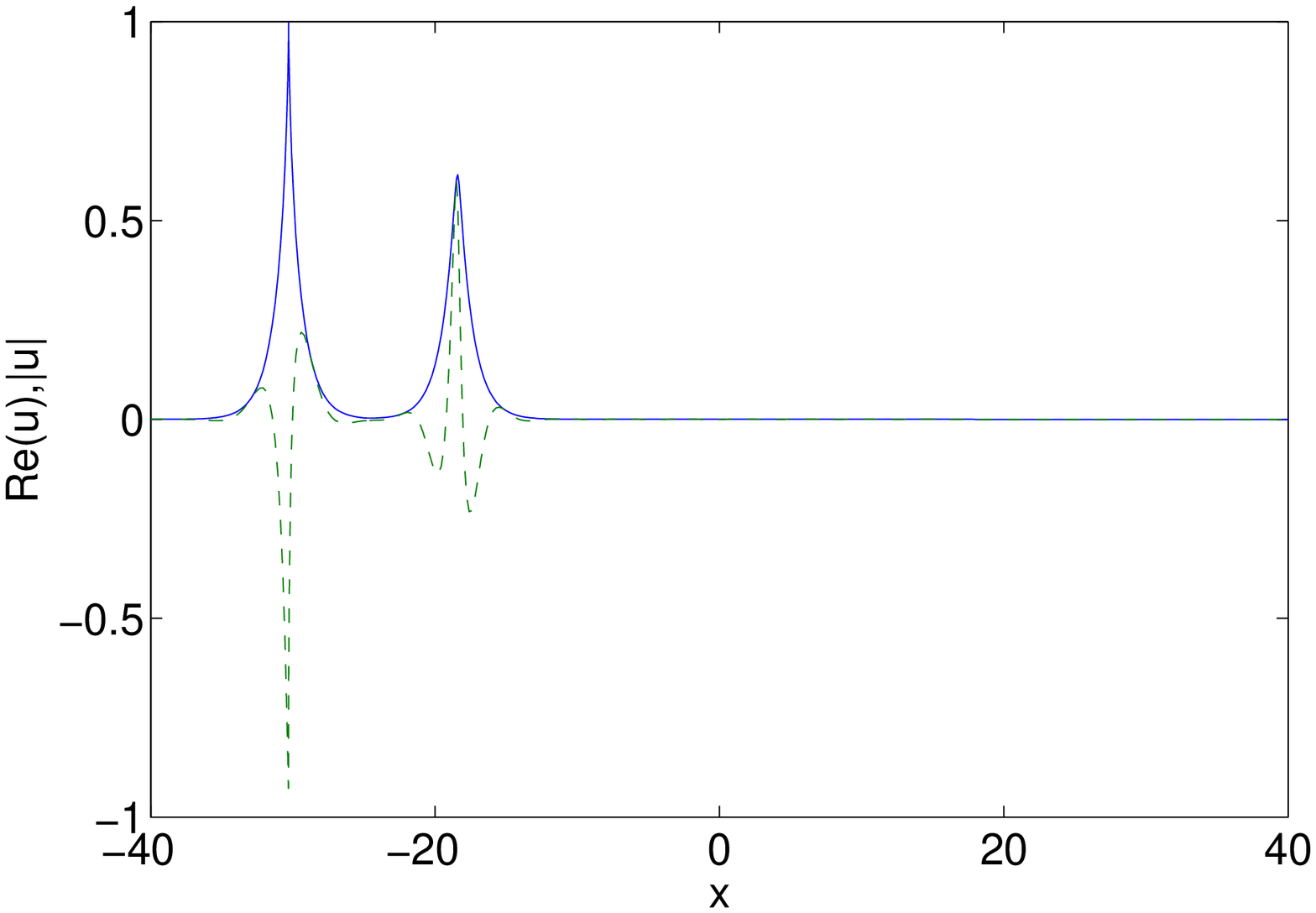}} \kern-0.3\textwidth
\hbox to
\textwidth{\hss(c)\kern0em\hss(d)\kern5em} \kern+0.25\textwidth
\caption{Two soliton solution for the semi-discrete complex short pulse
equation; solid line: $|u|$, dashed line: $Re(u)$. (a) $t=0$; (b) $t=45.0$;
(c) $t=60.0$; (d) $t=110.0$.}
\label{figure:smoothcupon}
\end{figure}
\section{Concluding Remarks}
%!
%We derived an integrable semi-discrete multi-component short pulse equation
We have derived an integrable semi-discrete analogue of the multi-component short pulse equation
proposed by Matsuno \cite{Matsuno_CSPE} based on a B\"acklund transform and Hirota's bilinear method.
%!
%Meanwhile,
We find its $N$-soliton solution in terms of pfaffians and prove it. Moreover, a complex short pulse equation, which possess smooth, cuspon or loop type envelop soliton, is proposed and its
semi-discrete analogue is constructed as well. We conclude the present paper by the
following remarks.

\begin{itemize}
\item The $N$-solution for multi-component short pulse equation given in the
present paper agrees with the one given by Matsuno in \cite{Matsuno_CSPE}.
This solution is a benchmark for the study of soliton interactions.

\item Similar to our previous results \cite{dCH,dCHcom,SPE_discrete1}, the
semi-discrete multi-component short pulse equation proposed here can be
served as an integrable numerical scheme, the so-called self-adaptive moving
mesh method, for the numerical simulation of multi-component short pulse
equation, as well as complex and coupled complex short pulse equation.
However, it is obviously beyond the scope of the present paper, we would
like to report our results in this aspect in a forthcoming paper.

\item The integrable fully discretization of the multi-component short pulse
equation is a further topic deserve to study.
\end{itemize}

\section*{Acknowledgments}
This work is partially supported by the National Natural Science Foundation of China (No. 11428102).


\begin{thebibliography}{100}
\bibitem{Kodamabook} A. Hasegawa, Y. Kodama, \emph{Solitons in Optical
Communications}, (Oxford University Press, 1995).

\bibitem{Agrawalbook} G.~P. Agrawal, \emph{Nonlinear Fiber Optics}, (Academic, San
Diego, 2001).

\bibitem{Boydbook} R.~W. Boyd, \emph{Nonlinear Optics}, (Academic Press, Boston,
1992).

\bibitem{Yarivbook} A. Yariv, P. Yeh, \emph{Optical Waves in Crystals: Propagation
and Control of Laser Radiation}, (Wiley-Interscience, 1983).

\bibitem{Zakharov} V.~E. Zakharov, A.~B. Shabat, JETP \textbf{34}  62--69 (1972).
%Eaxct theory of two-dimensional self-focusing and one-dimensional self-modulation of waves
%in nonlinear media,

\bibitem{Rothenberg} J.~E. Rothenberg,
 Opt. Lett. \textbf{17}  1340--1342 (1992).
%Space-time focusing: breakdown of the
%slowly varying envelope approximation in the self-focusing of femtosecond pulses,

\bibitem{SPE_Org} T. Sch\"afer  and C.~E. Wayne
  Physica D \textbf{196} 90--105 (2004).

\bibitem{SPE_CJSW} Y. Chung, C.~K.~R.T Jones , T. Sch\"afer and C.~E. Wayne %
 Nonlinearity \textbf{18} 1351--1374 (2005).

\bibitem{Robelo} M.~L. Robelo    Stud. Appl. Math. \textbf{81} 221--248 (1989).

\bibitem{Beals} R. Beals, M. Rabelo  and K. Tenenblat   Stud. Appl.
Math. \textbf{81} 125--151 (1989).

\bibitem{Sakovich} A. Sakovich  and S. Sakovich    J. Phys. Soc.
Jpn. \textbf{74} 239--241 (2005).

\bibitem{Brunelli1} J.~C.Brunelli  J. Math. Phys. \textbf{46} 123507 (2005).

\bibitem{Brunelli2} J.~C. Brunelli    Phys. Lett. A \textbf{353} 475--478 (2006).

\bibitem{Sakovich2} A. Sakovich  and S. Sakovich    J. Phys. A \textbf{39}
L361--367 (2006).

\bibitem{Kuetche} V.~K. Kuetche, T.~B. Bouetou and T.~C. Kofane    J. %
Phys. Soc. Jpn. \textbf{76} 024004 (2007).

\bibitem{Matsuno_SPE} Y. Matsuno   J. Phys. Soc. Jpn. \textbf{76} 084003 (2007).
\bibitem{Matsuno_SPEreview} Y. Matsuno, \emph{Handbook of Solitons:
Research, Technology and Applications} (edited by Lang S P and Bedore H, Nova, pp541, 2004).

\bibitem{Hirota} R. Hirota, \emph{The Direct Method in Soliton Theory},
(Cambridge University Press, 2004).

\bibitem{SPE_discrete1} B.-F Feng, K. Maruno and Y. Ohta    J. Phys.
A \textbf{43} 085203 (2010).

\bibitem{SPE_discrete2} B.-F Feng, J. Inoguchi, K. Kajiwara, K. Maruno and Y. Ohta
  J. Phys. A \textbf{44} 395201 (2011).

\bibitem{Kartashov} D.~V. Kartashov, A.~V. Kim and S.~A. Skobelev
JETP Letters \textbf{78} 276 (2003).

\bibitem{PKB_CSPE} M. Pietrzyk, I. Kanatt{\v s}ikov and U. Bandelow  %
 J. Nonl. Math. Phys. \textbf{15} 162-170 (2008).

\bibitem{Hoissen_CSPE} A. Dimakis, F. Muller-Hoissen   SIGMA \textbf{6}
055 (2010).

\bibitem{Matsuno_CSPE} Y. Matsuno   J. Math. Phys. \textbf{52} 123702 (2011).

\bibitem{Feng_CSPE} B.-F. Feng     J. Phys. A \textbf{45} 085202 (2012).

\bibitem{ZengYao_CSPE} Y. Yao and Y. Zeng     J. Phys. Soc. Jpn.
\textbf{80} 064004 (2011).

\bibitem{comSPE} B.-F. Feng,  Physica D \textbf{297} 62--75 (2015).

\bibitem{Brunelli_CSPE} J.~C. Brunelli, S. Sakovich   J. Math.
Phys. \textbf{54} 012701 (2013).

\bibitem{Hirota_Ohta_sG} R. Hirota, Y. Ohta    J. Phys. Soc. Jpn.
\textbf{60} 798 (1992).

\bibitem{JM} M. Jimbo, T. Miwa     Publ. Res. Inst. Math. Sci.
\textbf{19} 943-1001 (1983).

\bibitem{IwaoHirota} M. Iwao  and R. Hirota    J. Phys. Soc. Jpn.
\textbf{66} 577--588 (1997).

\bibitem{dcomSPE} B.~F. Feng, K. Maruno and Y. Ohta    Pacific
Journal of Mathematics for Industry,\textbf{6} 8 (2014).

\bibitem{dCH} Y. Ohta, K. Maruno  and B.~F. Feng    J. Phys. A \textbf{41}
355205 (2008).

\bibitem{dCHcom}  B.~F. Feng, K. Maruno and Y. Ohta   J. Comput.
Appl. Math \textbf{235} 229--243 (2010).
\end{thebibliography}
\end{document}